\documentclass[letterpaper,11pt]{article}
\usepackage[margin=1in]{geometry}
\usepackage{graphicx}
\usepackage{xcolor}
\usepackage{amsmath,amssymb}
\usepackage{amsthm}
\usepackage{mathtools}
\usepackage{enumitem}
\usepackage{algpseudocodex}
\usepackage{algorithm}
\usepackage{tikz}
\usepackage{verbatim}
\usepackage{subcaption}
\usepackage{url}
\usepackage{hyperref}

\newtheorem{theorem}{Theorem}
\newtheorem{lemma}[theorem]{Lemma}
\newtheorem{corollary}[theorem]{Corollary}
\theoremstyle{definition}
\newtheorem{definition}[theorem]{Definition}

\newtheorem{fact}[theorem]{Fact}
\newtheorem{remark}[theorem]{Remark}

\setlist[enumerate]{noitemsep, topsep=5pt, partopsep=0pt}
\setlist[itemize]{noitemsep, topsep=5pt, partopsep=0pt}

\title{Straightforward Entropy-Sensitive Mergesort}

\author{
  Bill Jin \\
  \small Department of Computer Science \\
  \small University of Toronto, Canada \\
  \small \texttt{bill.jin@mail.utoronto.ca}
  \and
  Alex Z. Xu \\
  \small Independent Researcher \\
  \small Hong Kong \\
  \small \texttt{xuzaxstk@gmail.com}
}
\date{}

\begin{document}

\begin{titlepage}
    \centering
    \vspace*{4cm}

    {\LARGE \bfseries Straightforward Entropy-Sensitive Mergesort \par}
    
    \vspace{2.5em}

    \begin{minipage}[t]{0.48\textwidth}
        \centering
        \textbf{Bill Jin}\\[0.3em]
        \small Department of Computer Science\\
        University of Toronto, Canada\\
        \texttt{bill.jin@mail.utoronto.ca}
    \end{minipage}
    \begin{minipage}[t]{0.48\textwidth}
        \centering
        \textbf{Alex Z. Xu}\\[0.3em]
        \small Independent Researcher\\
        Hong Kong\\
        \texttt{xuzaxstk@gmail.com}
    \end{minipage}

    \vspace{3em}
    
    \begin{minipage}{0.85\textwidth}
        \small
        \begin{center}\textbf{Abstract}\end{center}
        \vspace{-0.75em}
        In this paper, we present a stable mergesort variant, \textit{directional mergesort}, that to sort an array of $n$ elements makes no more than $nH+3n$ comparisons and $1.5nH+O(n)$ moves where $H$ is the run-based entropy of the input sequence, matching the best existing algorithms in run-adaptive sorting. However, our algorithm is surprisingly minimalistic: it leverages only skip checks, i.e., bypassing the merge step when both halves are already in order, and dynamically changing the direction of merging based on the state of the subarrays. As dynamic run scanning is avoided, all merge steps remain static and derivable from $n$, enabling a reduction to $O(1)$ words of stack space (i.e. space usage excluding the merge buffer) in \textit{directional mergesort\textsuperscript{++}}, thus improving over the predecessors' $O(\lg n)$ words. Importantly, as directional mergesort adapts only to non-decreasing runs, directional mergesort\textsuperscript{++} also applies a parallel set of rules for (strictly) decreasing runs, allowing the algorithm to also adapt to decreasing runs whilst retaining the original comparison bounds.
    \end{minipage}
    
    \vspace{2em}
    
    \vfill
\end{titlepage}

\section{Introduction}

Sorting is one of the most fundamental problems in computer science, carrying great practical and theoretical importance. For a comparison-based sorting algorithm to operate upon a sequence of $n$ elements, it is well known that the information theoretic lower bound is $\lg (n!) \approx n \lg n -1.4427n$ comparisons. Note that all logarithms throughout this paper are base 2.

However, the classical $\lg(n!)$ lower bound assumes an arbitrary input permutation where all $n!$ outcomes are equally likely. Conversely, real-world datasets frequently exhibit substantial pre-existing order. By exploiting this structure, adaptive sorting algorithms manage to break this generic $\Omega(n \log n)$ barrier, resulting in significant speedups in practical scenarios.

A recent advancement in \cite{MW18} has introduced two   algorithms that optimally adapt to pre-existing sorted \textit{runs} (i.e. maximal contiguous pre-sorted subarrays). They share running times of $O(nH + n)$ and use no more than $nH+3n$ comparisons, defined according to the run-based entropy $H=\sum_{i=1}^{r_\#} (r_i/n) \lg (n/r_i)$, where $r_i$ is the length of the $i$-th sorted run (non-decreasing or decreasing). However, these algorithms and earlier well-known approaches like Timsort (from \cite{Pet02}) rely heavily on run scanning and dynamically formulating run-dependent merge boundaries. This dynamic nature forces state tracking during execution, requiring $O(\lg n)$ words of stack space or sizable overhead to avoid an explicit stack (in \cite{GGS26}). Furthermore, the underlying topology of the \textit{merging tree} model for execution becomes highly volatile due to the reliance on runtime input rather than the static array index partitioning of classical mergesort, making proofs notoriously difficult. For example, the worst-case complexity of Timsort was finally proven in \cite{AJNP18}, 16 years after its initial release.

In this paper, we introduce an algorithm \textit{directional mergesort}, which uses no more than $nH+3n$ comparisons, $1.5nH+5n$ moves and $O(\lg n)$ words of auxiliary space excluding the buffers for merging, matching the comparison bound of Powersort and Peeksort from \cite{MW18} whilst retaining the simplicity of static merge boundaries, but does not adapt to decreasing runs. Additionally, we propose two optimizations: an extension of adaptive capabilities to decreasing runs and a reduction to $O(1)$ words of stack space. The resulting algorithm \textit{directional mergesort\textsuperscript{++}} achieves the final bounds of $nH+3n-r_\#$ comparisons, $1.5nH+6.5n$ moves and $O(1)$ words of auxiliary space excluding the buffers for merging, where $H$ is defined for both decreasing and non-decreasing runs.

\subsection{Technical Overview}
In Section 2, we present the algorithm of directional mergesort, which is a straightforward modification of a standard top-down mergesort. Its adaptivity relies solely on:
\begin{itemize}
    \item The well-known skip check that bypasses the merge when both halves are already in order.
    \item Dynamically altering the merge direction based on the state of the halves.
\end{itemize}
\noindent
In Section 3, we formally prove the claimed bounds of $nH+3n$ comparisons and $1.5nH+O(n)$ moves for the algorithm by analyzing the algorithm's behavior with the merging tree model. Crucially, all adaptation occurs locally within node operations and without altering the underlying merging tree geometry, giving a clean element-charging argument over static tree levels.

In Section 4, we then introduce two optimizations applied in directional mergesort\textsuperscript{++}:
\begin{itemize}
    \item Parallel mechanisms for decreasing runs that enable the algorithm to adapt to decreasing runs as well as non-decreasing runs with zero overhead in comparisons.
    \item An optimization that shrinks the top-down recursive call stack to $O(1)$ words, offering a distinct improvement over predecessors such as Powersort and Peeksort. This optimization also hinges on merge steps being static so they can be computed arithmetically, demonstrating the direct algorithmic dividends of simplicity beyond theoretical clarity.
\end{itemize}

\section{Directional Mergesort}

In this section, we will focus upon the elementary variant of directional mergesort. We derive \textsc{DirectionalMergesort} from standard top-down mergesort by making two simple adjustments: incorporating a standard skip check, and dynamically selecting the merge direction based on subarray pre-sortedness. Notably, these adjustments do not affect the boundaries of recursive calls.

\begin{algorithm}
\caption{Directional Mergesort}
\begin{algorithmic}[1]
\Function{Sort}{$A[], \ell, r$}
    \Comment{Processing Subarray $A[\ell..r)$}
    \If{$r-\ell < 2$}
        \Return{\texttt{true}}
    \EndIf
    \State $m \gets \lfloor (\ell+r)/2 \rfloor$
    \State $a \gets \Call{Sort}{A,\ell,m}$
    \State $b \gets \Call{Sort}{A,m,r}$
    \If{$A[m - 1] \le A[m]$}
        \Return{$a \land b$} \Comment{Skip Check}
    \EndIf
    \If{$b$}
        \Call{MergeBw}{$A,\ell,m,r$} \Comment{Directional Merging}
    \Else
        \ \Call{MergeFw}{$A,\ell,m,r$}
    \EndIf
    \State \Return{\texttt{false}}
\EndFunction
\end{algorithmic}
\end{algorithm}
\noindent
The return value indicates whether the processed subarray is \textit{pure} (i.e., entirely pre-sorted). Hence, base cases return \texttt{true}, whereas any subarray requiring an explicit merge returns \texttt{false}.
\paragraph{Skip Check.} Before executing a merge, the algorithm evaluates $A[m-1] \le A[m]$. If satisfied, the adjacent subarrays are already in sorted order. The algorithm skips the merge step entirely and returns $a \land b$, propagating structural purity up the call tree iff both children were pure.

\paragraph{Directional Merging.}
When an active merge is required, the algorithm inspects the purity flag $b$ of the right child $A[m..r)$ to decide the direction of the merge.
\begin{itemize}
    \item \textbf{Backward Merge ($b = \texttt{true}$):} A purity flag of $b = \texttt{true}$ indicates that $A[m,r)$ is part of a sorted run. The directional merge optimizes for the case when $A[\ell,m)$ contains parts from the same sorted run. Because the continuous run is non-decreasing, all run elements in the right half are greater than or equal to all run elements in the left half. When sorting $A[\ell..m)$, larger non-run elements may move to index $m-1$, failing the skip check ($A[m-1] > A[m]$). By invoking \textsc{MergeBw}, the merge repeatedly selects the larger element (favoring the right half on ties). This ensures the entire right half $A[m,r)$ is consumed before selecting any left-half run elements. Once $A[m,r)$ is exhausted, the remaining elements in $A[\ell,m)$ are simply appended to the end of the merged result using 0 comparisons.
    \item \textbf{Forward Merge ($b = \texttt{false}$):} When $b = \texttt{false}$, the right child is not a pure run. The algorithm defaults to \textsc{MergeFw}, merging conventionally from left to right (smallest to largest). This offers a symmetric optimization: if the left half was pure ($a = \texttt{true}$) the trailing elements of the run in the right half would be appended using 0 comparisons.
\end{itemize}

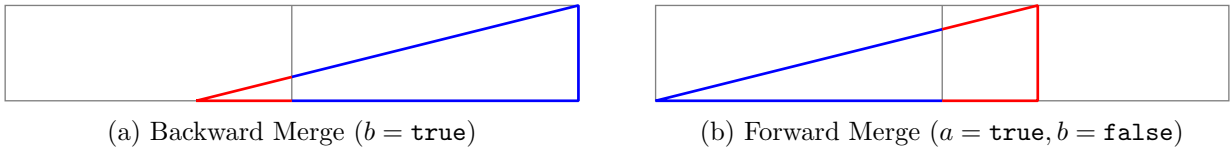
\begin{figure}[htbp]
    \centering
    \begin{subfigure}[b]{0.48\textwidth}
        \centering
        \resizebox{\linewidth}{!}{
            \begin{tikzpicture}
                \draw[gray, thin] (0,0) rectangle (6,1);
                \draw[gray, thin] (3,0) -- (3,1);
                \draw[blue, thick, line join=bevel] (3,0) -- (6,0) -- (6,1) -- (3,1/4);
                \draw[red, thick, line join=bevel] (3,0) -- (2,0) -- (3,1/4);
            \end{tikzpicture}
        }
        \caption{Backward Merge ($b = \texttt{true}$)}
        \label{fig:merge_fw_savings}
    \end{subfigure}
    \hfill
    \begin{subfigure}[b]{0.48\textwidth}
        \centering
        \resizebox{\linewidth}{!}{
            \begin{tikzpicture}
                \draw[gray, thin] (0,0) rectangle (6,1);
                \draw[gray, thin] (3,0) -- (3,1);
                \draw[blue, thick, line join=bevel] (3,0) -- (0,0) -- (3,3/4);
                \draw[red, thick, line join=bevel] (3,0) -- (4,0) -- (4,1) -- (3,3/4);
            \end{tikzpicture}
        }
        \caption{Forward Merge ($a = \texttt{true}, b = \texttt{false}$)}
        \label{fig:merge_bw_savings}
    \end{subfigure}
    \caption{Subarray configurations exhibiting directional merge savings.}
    \label{fig:directional_savings}
\end{figure}

\section{Complexity Analysis}
In this section we prove the claimed bounds of directional mergesort by analyzing the merging tree.

\begin{definition}[Merging Tree]
    The merging tree serves as a model for the execution of the algorithm, which takes the exact shape of a segment tree over the array of $n$ elements. The root of the tree at level 0 represents the entire array of size $n$. For any internal node $v$ covering the subarray $A[\ell,r)$ at level $i$, its children at level $i+1$ are defined by the algorithm's recursive splitting logic: a left child $L$ covering $A[\ell,m)$ and a right child $R$ covering $A[m,r)$ where $m=\lfloor (\ell+r)/2 \rfloor$. The leaves correspond to singular elements of size $1$, which is the base-case of the recursive algorithm. At every internal node, a merge operation combines the sorted subarrays from its left and right children.
\end{definition}

\begin{fact}
    At any level $h$ in the merging tree, the size of nodes are either $\lfloor n/2^h \rfloor$ or $\lceil n/2^h \rceil$.
\end{fact}

\begin{fact}
    At any level $h$ in the merging tree, all node subarrays are contiguous and pairwise disjoint.
\end{fact}

\subsection{Comparison Cost Analysis}
\noindent
We first isolate the skip check comparisons, which are guaranteed to occur before the merge at any internal node, so that we may focus on analyzing solely the cost of merging subarrays.

\begin{lemma}
    The total cost of skip checks across the entire algorithm is exactly $n-1$ comparisons.
\end{lemma}

\begin{proof}
    At every internal node, exactly $1$ comparison is used for the skip check. Because the merging tree is a full binary tree, the number of internal nodes is exactly one less than the number of leaves. Thus, a total of $n-1$ skip checks occur, requiring a total of $n-1$ comparisons.
\end{proof}

\noindent
\textbf{Accounting Scheme:} Instead of assigning the cost of a merge to the internal node, we distribute the cost to the individual elements.
For each comparison in a merge loop, the winning element is charged with the $1$ comparison. After exhausting one side, no comparisons are made or charged.
We now consider a pre-existing non-decreasing run of length $r$ at the bottom level of the tree, with the goal of bounding the merge charges applied specifically to the elements of this run. We categorize the internal nodes relative to this run into four cases:

\begin{enumerate}
    \item Fully Contained (Case 1): The node's subarray consists entirely of elements from the run. The merge is guaranteed to be skipped (costing $0$).
    \item Boundary (Case 2): The node's subarray contains run elements as either a prefix or suffix.
    \item Encapsulating (Case 3): The node's subarray completely encapsulates the entire run, containing non-run elements on both the left and right sides.
    \item Outside (Case 4): The node's subarray contains no elements from the run (costing $0$).
\end{enumerate}

\noindent
To bound the total merge charges incurred by a run of length $r$, we partition the levels of the merging tree into two distinct regimes using the threshold depth $t = \lceil \lg(n/r) \rceil$: for a node at level $h$, it is in the upper levels if $h<t$, and it is in the lower levels if $h\geq t$.

\begin{lemma}[Comparisons at Upper Levels]
    For a non-decreasing run of length $r$, the total number of comparisons charged to its elements at levels $h < t$ is bounded by $rt$.
\end{lemma}

\begin{proof}
    At any level $h$, the nodes form a partition of the array into disjoint contiguous subarrays (Fact 3). Consequently, each element in the array belongs to exactly one node per level.

    At any upper level $h < t = \lceil \lg(n/r) \rceil$, the node size is $\lfloor n/2^h \rfloor \geq r$ according to Fact~2. Because node size exceeds $r/2$, Case~1 nodes cannot be guaranteed to exist in this regime.
    
    Hence, we assume no savings occur. Each of the $r$ elements in the run is charged $1$ comparison per level, incurring a collective charge of at most $r$ comparisons at each level. Summing across the $t$ upper levels $h = 0, \dots, t-1$, the total merge cost charged to the run is bounded by $rt$.
\end{proof}

\begin{lemma}[Comparisons at Lower Levels]
    For a non-decreasing run of length $r$, the total number of comparisons charged to its elements at levels $h \geq t$ is bounded by $2(n/2^t)$.
\end{lemma}

\begin{proof}
    For levels $h \geq t$ we have that $h \geq \lceil \lg(n/r) \rceil \geq \lg (n/r)$ which gives $n/2^h \leq r$. However, $r$ is an integer so $\lceil n/2^h \rceil \leq r$, bounding the node sizes on these levels at $r$ elements. Thus, Case~3 (Encapsulating) nodes, which require at least $r + 2$ elements, cannot exist on these levels. Therefore, any node relative to the run must either be fully contained within it (Case~1), fall completely outside of it (Case~4), or intersect exactly one of the run's edges (Case~2). The only such nodes with costs relevant to the run are Case 2 (Boundary) nodes. Because the run spans a contiguous interval, there are at most two Case 2 nodes at any level: one containing the left boundary of the run, and one containing the right boundary.

    \paragraph{Left Boundary Analysis:} Let $v_i$ be the Case 2 node on the left boundary at level $i \ge t$, with size $n_i = \vert{}v_i\vert{}$. Because $v_i$ contains the left boundary of the run, the run forms a suffix of $v_i$. Node $v_i$ splits into children $L$ and $R$ where $|L|=\lfloor n_i/2 \rfloor$ and $|R|=\lceil n_i/2 \rceil$. We trace the path of left boundary nodes down the tree:

    \begin{enumerate}
        \item \textbf{Subcase A (Boundary falls in $R$):} The left child $L$ lies entirely outside the run, and the run elements in $v_i$ fall within $R$. We explicitly define the next boundary node in our path as $v_{i+1} = R$. The maximum comparisons charged is $1$ for each of the $s_i$ run elements in $v_i$. Since the run boundary falls inside $R$, $s_i \le \vert{}R\vert{} - 1$. Given $\vert{}R\vert{} = \lceil n_i / 2 \rceil$, it follows that:
        $$s_i \le \lceil n_i/2 \rceil -1 \leq \lfloor n_i / 2 \rfloor = \vert{}L\vert{} = n_i - \vert{}R\vert{} = n_i - n_{i+1}$$
        \item \textbf{Subcase B (Boundary falls in $L$):} The right child $R$ is entirely contained within the run, making $R$ a pure sub-run. We explicitly define the next boundary node in our path as $v_{i+1} = L$. Because $R$ is pure, its recursive sort returns \texttt{true}, setting $b \gets \texttt{true}$. Furthermore, because $L$ contains a boundary between two runs, it is not pure, and its recursive sort returns \texttt{false}, setting $a \gets \texttt{false}$. The directional mergesort thus selects \textsc{mergeBw}, which compares elements from largest to smallest (selecting from $R$ upon ties). Because all elements in $R$ and $L \cap \text{Run}$ belong to the same sorted run, every element in $R$ is $\ge$ all run elements in $L$. Consequently, the algorithm exhausts $R$ inside the main comparison loop (costing $\vert{}R\vert{}$ comparisons) before selecting any elements from $L \cap \text{Run}$. The remaining run elements in $L$ are merged direct using $0$ comparisons. Thus, the charged cost is exactly $\vert{}R\vert{}$. Since $n_i = \vert{}L\vert{} + \vert{}R\vert{}$ and $v_{i+1} = L$, the cost is:
        $$\vert{}R\vert{} = n_i - \vert{}L\vert{} = n_i - n_{i+1}$$
        \item \textbf{Subcase C (Boundary aligns perfectly):} If the boundary aligns perfectly with the partition between $L$ and $R$ at some level $k$, $R$ is a pure sub-run (Case~1) and $L$ contains zero run elements (Case~4). Node $v_k$ is therefore the final Case~2 node in the boundary path. As in Subcase B, $R$ is pure, triggering \textsc{mergeBw} which exhausts the elements in $R$ at a cost of $\vert{}R\vert{}$ comparisons (using worst-case charge of $1$ comparison per element). To cleanly terminate the telescoping sum, we define the final term $n_{k+1} = \vert{}L\vert{}$. The cost at this terminal level $k$ is:
        $$\vert{}R\vert{} = n_k - \vert{}L\vert{} = n_k - n_{k+1}$$
    \end{enumerate}

    \noindent
    In all cases, the comparison cost charged to the run at any level $i$ along the left boundary path is bounded by $n_i - n_{i+1}$. Summing these costs from the top of the lower-level regime (level $t$) down to the termination level $k$ yields a telescoping sum. Because $n_{k+1} \ge 1$ for any node in the tree, and $n_t\leq \lceil n/2^t \rceil$ according to Fact 2, the total left boundary cost is bounded by:
    $$\sum_{i=t}^{k} (n_i - n_{i+1}) = n_t - n_{k+1} \le n_t - 1 \leq \lceil n/2^t \rceil -1 \leq n/2^t$$

    \paragraph{Right Boundary Analysis:} By symmetry, the right boundary yields an identical bound. Because the run forms a prefix of any Case~2 node $v'_i$ (with left child $L$ and right child $R$) on the right boundary, the algorithm mirrors the left boundary logic. When the boundary falls in the left child, applying the worst-case charge of $1$ comparison per element gives $s'_i \leq n'_i-n'_{i+1}$ comparisons. If it falls into the right child, the left child is pure giving $a = \texttt{true}$ whilst the right child is impure as it contains a boundary between runs which are maximal, giving $b=\texttt{false}$. The directional mergesort then utilizes \textsc{mergeFw} to efficiently exhaust the pure left child first, ensuring the comparison cost charged at level $i$ is bounded by $n'_i - n'_{i+1}$. At the right termination level $k'$, the boundary perfectly aligns with the partition between the children $L$ and $R$, meaning $L$ is pure and $R$ contains no run elements. Charging $1$ comparison to every element in $L$, the cost at this terminal level is $|L|=n'_{k'}-|R|=n'_{k'}-n'_{k'+1}$ (defining $n'_{k'+1}=|R|$ to terminate the telescoping sum). Summing this telescoping sum across the lower levels yields the exact same overall bound of $n/2^t$ for the right boundary. Thus, the total cost charged at the lower levels are:
    \begin{equation*} n/2^t+n/2^t=2(n/2^t) \end{equation*}
    If on one (or both) ends no boundary nodes exist, then the cost is trivially $0\leq n/2^t$, hence the upper bound of $2(n/2^t)$ remains valid.
\end{proof}

\begin{lemma}
    For a non-decreasing run of length $r$, the total number of merge comparisons charged to its elements across all levels of the merging tree is bounded by $r \lg(n/r) + 2r$.
\end{lemma}

\begin{proof}
    By partitioning the tree at threshold depth $t = \lceil \lg(n/r) \rceil$, we sum the upper-level costs (Lemma 5) and lower-level costs (Lemma 6). The total comparison charge is bounded by:
    $$r \lceil \lg(n/r) \rceil + 2(n/2^t)$$
    To simplify this, let $\gamma = \lceil \lg(n/r) \rceil - \lg(n/r)$ represent the fractional ceiling overhead, where $0 \le \gamma < 1$. We can rewrite the sum in terms of $\gamma$:
    $$r(\lg(n/r) + \gamma) + 2\left(n/2^{\lg(n/r)+\gamma}\right)$$
    \begin{equation}= r \lg(n/r) + r(\gamma+2^{1-\gamma})\end{equation}

    \noindent
    As $2^{-\gamma}$ is convex, we can apply the secant line bound to $2^{-\gamma}$ for $\gamma \in [0, 1]$:
    \begin{align*}
    2^{-\gamma} &\leq 1-(1/2)\gamma \\[0.25em]
    2(2^{-\gamma}) &\leq 2-\gamma \\[0.25em]
    2^{1-\gamma}+\gamma &\leq 2
    \end{align*}

    \noindent
    Substituting this upper limit into (1) yields the final bound:
    \begin{equation*} r \lg (n/r)+2r \qedhere \end{equation*}
\end{proof}

\begin{theorem}[Total Comparison Complexity]
    An array of size $n$ containing $r_\#$ non-decreasing runs of lengths $r_1, r_2, \dots, r_{r_\#}$ is sorted in at most $nH + 3n$ comparisons, where the run-based entropy $H = \sum_{i=1}^{r_\#} (r_i/n) \lg \left(n/r_i\right)$.
\end{theorem} 

\begin{proof}
The total number of comparisons from merging is derived by summing the per run cost (Lemma 7) across all $r_\#$ sorted runs (with lengths $r_1, r_2, \dots, r_{r_\#}$ where $\sum r_i = n$):
$$\left( \sum_{i=1}^{r_\#} r_i \lg (n/r_i) +2r_i \right)=\left( \sum_{i=1}^{r_\#} r_i \lg (n/r_i)\right)+2n=n\left( \sum_{i=1}^{r_\#} (r_i/n) \lg (n/r_i)\right)+2n=nH+2n$$
\noindent
Finally, we reintroduce the cost of skip checks from Lemma 4:
\begin{equation*} nH+2n+(n-1) = nH+3n-1 \leq nH+3n \qedhere \end{equation*}
\end{proof}

\subsection{Data Movement Analysis}

To analyze the data movement complexity, we apply a similar structural framework based on the merging tree established in Section 3.1.

The algorithm utilizes a half-buffered merge. In an internal node $v$, the algorithm first copies either the left child $L$ or right child $R$ into the buffer, moving at most $\max(|L|, |R|) \le \lceil |v|/2 \rceil$ elements by Fact 2. It then uses at most $m$ moves to form the merged result in the original array. For $|v| \in \mathbb{Z}$, total moves are bounded by:
    \begin{equation*} |v| + \lceil |v|/2 \rceil \leq |v|+|v|/2+1/2 = 1.5|v| + 1/2 \end{equation*}

\noindent
Because the merging tree contains exactly $n-1$ internal nodes (see Lemma 4), the $+1/2$ overhead term incurred by each merge at every internal node contributes an absolute maximum of $(n-1)/2$ moves across the entire execution. We globally separate this $O(n)$ overhead, leaving an aggregate move pool of exactly $1.5|v|$ for each internal node $v$. 

\paragraph{Amortized Accounting Method:} We distribute the $1.5|v|$ move pool evenly among the $|v|$ elements in $v$, charging exactly $1.5$ moves per element. The key difference from the approach for comparisons is that this $1.5$ charge is a uniform mathematical abstraction over the node's capacity rather than a physical count of whether an element was copied to the buffer or written as a residual (whereas Section 3.1 uses the actual cost of selecting the element in the merge). Hence, local deductions for uncopied residuals cannot be applied to this charge. However, if a merge is skipped entirely in Case 1, then $0$ moves are made, hence the per element cost is also $0$. Similarly, Case 4 nodes contain no elements from the run, thus the charge to the run is also $0$.

Equipped with this charging rate, we bound the aggregate move cost for a run of length $r$, utilizing the same two regime partitioning established in Section~3.1. However, to tighten the move complexity bound, we shift the regime threshold to $t+1$.

\begin{lemma}[Moves at Upper Levels]
\label{lemma:moves_upper}
    For a non-decreasing run of length $r$, the total move charge incurred at upper levels ($h < t+1$) is bounded by $1.5r(t+1)$.
\end{lemma}

\begin{proof}
    Paralleling Lemma 5, each of the $r$ elements in the run belongs to exactly one node per level and is charged $1.5$ moves. Across all $t+1$ levels, the total move charge is $1.5r(t+1)$.
\end{proof}

\begin{lemma}[Moves at Lower Levels]
\label{lemma:moves_lower}
    For a non-decreasing run of length $r$, the total move charge incurred at lower levels ($h \ge t+1$) is bounded by $3(n/2^t)$.
\end{lemma}

\begin{proof}
    As previously established (in the proof of Lemma 6), at levels $h \ge t+1> t$, Case~1 and Case~4 nodes execute no moves (relevant to the run), and Case~3 nodes cannot exist. Thus, the moves at lower levels are purely from at most two Case~2 boundary nodes at each level.
    
    By Fact 2, the maximum size of any node at level $h$ is $\lceil n/2^h \rceil$.  Consequently, the maximum number of run elements contained in a boundary node at level $h$ is at most $\lceil n/2^h \rceil-1\leq n/2^h$. Across both boundary nodes at level $h$, the total number of run elements charged is at most $2(n/2^h)$. Applying the uniform accounting charge of $1.5$ moves per element, the total charge at level $h$ is bounded by $3(n/2^h)$. Summing this geometric series for all lower levels:
    \begin{equation*} \sum_{h=t+1}^{\lg n} 3(n/2^h) \leq 3 \left( \frac{n/2^{t+1}}{1 - 1/2} \right) = 6(n/2^{t+1})=3(n/2^t) \qedhere \end{equation*}
\end{proof}

\begin{theorem}[Total Move Complexity]
\label{theorem:total_moves}
    An array of size $n$ containing $r_\#$ pre-existing sorted runs is sorted in at most $1.5nH + 5n$ moves, where $H = \sum_{i=1}^{r_\#} (r_i/n) \lg \left(n/r_i\right)$ is the run-based entropy.
\end{theorem}

\begin{proof}
    Using the threshold $t+1 = \lceil \lg(n/r) \rceil+1 = \lg(n/r) + \gamma+1$, we combine the upper (Lemma 9) and lower (Lemma 10) level costs for a single run of length $r$:
    \begin{equation} 1.5r(t+1)+3(n/2^t) = 1.5 r \lg(n/r)+r\left(1.5+1.5\gamma+3\left(2^{-\gamma}\right)\right) \end{equation}
    \noindent
    Rearranging the secant line bound on $2^{-\gamma}$ for $\gamma \in[0,1]$:
    \begin{align*}
        2^{-\gamma} &\leq 1-(1/2)\gamma \\[0.25em]
        3(2^{-\gamma}) &\leq 3-1.5\gamma \\[0.25em]
        3(2^{-\gamma})+1.5\gamma &\leq 3
    \end{align*}

    \noindent
    Substituting into (2), the moves made to process a run of length $r$ is given by:
    $$1.5 r \lg(n/r) + r\left(1.5+1.5\gamma+3\left(2^{-\gamma}\right)\right) \leq 1.5r \lg (n/r) + 4.5r$$
    
    \noindent
    Summing this bound across all $r_\#$ sorted runs (with lengths $r_1, r_2, \dots, r_{r_\#}$ where $\sum r_i = n$):
    $$\sum_{i=1}^{r_\#} \left( 1.5 r_i \lg(n/r_i) + 4.5 r_i \right) = 1.5n \left( \sum_{i=1}^{r_\#} (r_i/n) \lg(n/r_i) \right) + 4.5n = 1.5nH + 4.5n$$
    
    \noindent
    Lastly, we reintegrate the $(n-1)/2$ overhead moves isolated from the initial merge cost:
    \begin{equation*} 1.5nH + 4.5n + (n-1)/2 \leq 1.5nH + 5n \qedhere \end{equation*}
\end{proof}

\subsection{Summary}
\begin{corollary}[Overall Complexity]
    Directional mergesort sorts an array of size $n$ with run-based entropy $H$ using no more than $nH + 3n$ comparisons and $1.5nH + 5n$ element moves. Excluding the merge buffer of size $\lceil n/2 \rceil$, the algorithm uses $O(\lg n)$ auxiliary words.
\end{corollary}

\begin{proof}
    Comparison and move bounds follow immediately from Theorem 8 and Theorem 11. The $O(\lg n)$ auxiliary words accounts for the recursive call stack, which is bounded by the merging tree's height of $\lceil \lg n \rceil$ due to the balanced binary splitting logic of the algorithm.
\end{proof}

\begin{remark}[Scope of Run Entropy]
\label{rem:non_decreasing_runs}
    The run-based entropy $H = \sum (r_i/n) \lg(n/r_i)$ in Corollary 12 is defined exclusively over non-decreasing pre-existing runs. If the input contains strictly decreasing runs, they are treated as individual runs of length $1$.
\end{remark}

\section{Optimized Variants}
Now that the complexity bounds of directional mergesort are established, we proceed to consider two optimizations that result in directional mergesort\textsuperscript{++}.

\subsection{Adapting to Decreasing Runs}

In the base algorithm, the returned Boolean flag tracks if the subarray was purely non-decreasing (\texttt{true} iff so) according to the success of skip checks. However, as skip checks only verify if the two halves are already in non-decreasing order, decreasing runs are treated as individual runs of length $1$. Thus, we extend the return value to two flags $x_1$ and $x_2$, where $x_1$ is the standard non-decreasing flag and $x_2$ tracks if the subarray was purely strictly decreasing (\texttt{true} iff so). In practice, these flags can be represented with a 2-bit unsigned integer bitmask for simpler implementation.

\begin{algorithm}
\caption{Directional Mergesort\textsuperscript{+1}}
\begin{algorithmic}[1]
\Function{Sort}{$A[], \ell, r$}
    \Comment{Processing Subarray $A[\ell..r)$}
    \If{$r-\ell < 2$}
        \Return{\texttt{(true, true)}}
    \EndIf
    \If{$r-\ell = 2$}
        \If{$A[\ell] \le A[\ell+1]$}
            \Return{\texttt{(true, false)}}
        \Else
            \ \Return{\texttt{(false, true)}}
        \EndIf
    \EndIf
    \State $m \gets \lfloor (\ell+r)/2 \rfloor$
    \State $(a_1,a_2) \gets \Call{Sort}{A,\ell,m}$
    \State $(b_1,b_2) \gets \Call{Sort}{A,m,r}$
    \If{$a_2 \land b_2$}
        \If{$A[m - 1] > A[m]$}
            \Return{\texttt{(false, true)}}
            \Comment{Decreasing Skip Check}
        \EndIf
    \EndIf
    \If{$a_2$} \Call{Reverse}{$A,\ell,m$}
        \Comment{Reversing Decreasing Halves}
        \EndIf 
    \If{$b_2$} \Call{Reverse}{$A,m,r$}
        \EndIf
    \If{$\text{\textbf{not}} (a_2 \land b_2)$}
        \If{$A[m - 1] \le A[m]$}
            \Return{$\texttt{($a_1 \land b_1$, false)}$}
            \Comment{Standard Skip Check}
        \EndIf
    \EndIf
    \If{$b_1 \lor \lnot( b_2 \lor a_1) $}
        \Call{MergeBw}{$A,\ell,m,r$}
        \Comment{Directional Merging}
    \Else
        \ \Call{MergeFw}{$A,\ell,m,r$}
    \EndIf
    \State \Return{\texttt{(false, false)}}
\EndFunction
\Procedure{DirectionalMergesort\textsuperscript{+1}}{$A[],n$}
    \If{\Call{Sort}{$A,0,n$} $=$ \texttt{(false, true)}}
        \Call{Reverse}{$A,0,n$}
    \EndIf
\EndProcedure
\end{algorithmic}
\end{algorithm}

\noindent
This variant has two separate base cases. When the subarray has $1$ element, it is vacuously both non-decreasing and strictly decreasing, so the procedure returns both as \texttt{true}. When the subarray size is 2, one comparison suffices to determine the relationship between the two elements, allowing us to concretely decide if the subarray is non-decreasing or strictly decreasing.

To avoid unnecessary data movement and allow for the maximal extension of decreasing runs, physical reversals are deferred: a purely decreasing subarray is reversed only when its parent subarray is not purely decreasing where the decreasing flag is dropped. To handle the edge case where the entire array is in strictly decreasing order and not corrected during any recursive calls, the outer procedure \textsc{DirectionalMergesort\textsuperscript{+1}} intercepts this return value and performs a single top-level reversal to reorient the array in sorted order. Note that stability is maintained by the reversals as they only operate on strictly decreasing subarrays that cannot contain equal elements, ensuring the relative order of equal elements are not inverted.

\paragraph{Skip Check.} To adapt to decreasing sequences, the algorithm introduces a decreasing skip check. If both halves are strictly decreasing, the check evaluates if the two decreasing child subarrays seamlessly connect into a larger strictly decreasing run, skipping the merge step and propagating this information up the tree. This ensures that any Case 1 nodes fully contained within a decreasing run bypass the merge, costing zero comparisons. If any of the halves is not decreasing, the algorithm utilizes the standard skip check ($A[m-1] \le A[m]$) to evaluate non-decreasing purity instead. Because the flag conditions for executing the two skip check are mutually exclusive, at most $1$ comparison is performed, hence the global overhead for skip checks remains at $n-1$ comparisons.

\paragraph{Directional Merging.} The directional merging logic follows from the elementary variant. However, \textsc{MergeBw} is also invoked when neither the right half is purely decreasing ($b_2=\texttt{false}$) nor the left half is purely non-decreasing ($a_1=\texttt{true}$), utilizing the backwards merge when there are no potential forward merge savings. This leads to additional savings for decreasing runs:

\begin{itemize}
    \item \textbf{Backward Merge:} When the left half is fully contained by a decreasing run that extends into the right half, all run elements in the right half are strictly less than the elements in the left half. In this situation, $a_2=\texttt{true}$, $b_2=\texttt{false}$ (as it contains a boundary), and $a_2=\texttt{true}$ implies $a_1=\texttt{false}$ when $|L|>1$ ($|L|=1$ handled in Appendix A.1), triggering \textsc{MergeBw}. The backwards merge then exhausts the left half before selecting any run elements from the right half. These remaining run elements are directly appended using $0$ comparisons.
    
    \item \textbf{Forward Merge:} Otherwise, when the right half is a decreasing sub-run $b_2=\texttt{true}$ which implies $b_1=\texttt{false}$ as $|R|>1$ (when $r-\ell\geq3$, which holds for all non base cases), \textsc{MergeFw} is called. If the run begins in the left half, this gives a symmetrical saving where the right half is exhausted first, and the run elements in the left half are appended with no comparisons.
\end{itemize} 

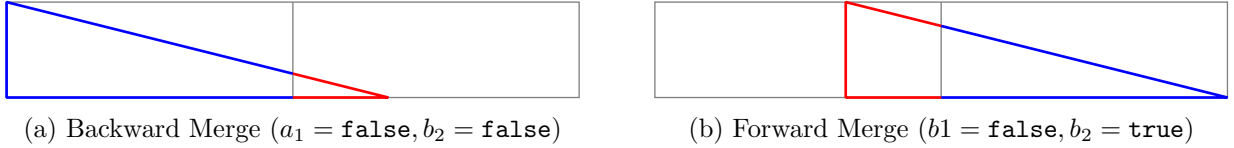
\begin{figure}[htbp]
    \centering
    \begin{subfigure}[b]{0.48\textwidth}
        \centering
        \resizebox{\linewidth}{!}{
            \begin{tikzpicture}
                \draw[gray, thin] (0,0) rectangle (6,1);
                \draw[gray, thin] (3,0) -- (3,1);
                \draw[blue, thick, line join=bevel] (3,0) -- (0,0) -- (0,1) -- (3,1/4);
                \draw[red, thick, line join=bevel] (3,0) -- (4,0) -- (3,1/4);
            \end{tikzpicture}
        }
        \caption{Backward Merge ($a_1 = \texttt{false},b_2=\texttt{false}$)}
        \label{fig:merge_bw_savings_rev}
    \end{subfigure}
    \hfill
    \begin{subfigure}[b]{0.48\textwidth}
        \centering
        \resizebox{\linewidth}{!}{
            \begin{tikzpicture}
                \draw[gray, thin] (0,0) rectangle (6,1);
                \draw[gray, thin] (3,0) -- (3,1);
                \draw[blue, thick, line join=bevel] (3,0) -- (6,0) -- (3,3/4);
                \draw[red, thick, line join=bevel] (3,0) -- (2,0) -- (2,1) -- (3,3/4);
            \end{tikzpicture}
        }
        \caption{Forward Merge ($b1 = \texttt{false}, b_2 = \texttt{true}$)}
        \label{fig:merge_fw_savings_rev}
    \end{subfigure}
    \caption{Subarray configurations exhibiting directional merge savings on decreasing runs.}
    \label{fig:directional_savings_rev}
\end{figure}

\noindent
With the mechanics of directional mergesort\textsuperscript{+1} established, we now prove that it achieves the same optimal comparison bounds for both non-decreasing and decreasing sequences.

\begin{lemma}[Merge Comparisons for Non-Decreasing Runs]
    For a non-decreasing run of length $r$, the total number of merge comparisons charged to its elements across all levels of the merging tree in directional mergesort\textsuperscript{+1} is bounded by $r \lg(n/r) + 2r$.
\end{lemma}

\begin{proof}
    Within any purely non-decreasing run, $A[i] \le A[i+1]$ holds for all
    pairs of elements from the run. Hence, any size 2 base case fully contained
    by the run identifies the non-decreasing order and returns $x_2 = 0$,
    resulting in unchanged skip check and directional merging savings upwards the merging tree (full details in Appendix A.1).

    For the two elements at the ends of the run, they may be attributed as decreasing together with adjacent non-run elements if they are in Case 2 nodes. Since such a node contains non-run elements, $x_1 = 0$ is correctly identified and the directional merge remains valid. Moreover, the worst-case accounting of Lemma 5 and Lemma 6 already assumes merges over mixed regions are never fortuitously skipped (i.e. elements from different runs happening to be in the correct order resulting in a skip), so the altered execution of the skip check does not affect the bound. Hence, the bound in Lemma 7 is retained.
\end{proof}

\begin{lemma}[Comparisons at Upper Levels: Decreasing]
    For a decreasing run of length $r$, the total number of comparisons charged to its elements at levels $h < t$ is bounded by $rt$.
\end{lemma}

\begin{proof}
    By the accounting scheme in Lemma 5, in the worst-case scenario, each of the $r$ elements in the run is charged the worst-case of $1$ comparison per level. Across all $t$ levels ($h=0, \dots, t-1$), the total charge to the run is at most $rt$.
\end{proof}

\begin{lemma}[Comparisons at Lower Levels: Decreasing]
    For a decreasing run of length $r$, the total number of comparisons charged to its elements at levels $h \geq t$ is bounded by $2(n/2^t)$.
\end{lemma}

\begin{proof}
    As established in the proof of Lemma 6, at levels $h\geq t$, nodes are too small to fully encapsulate the run, thus Case 3 nodes cannot exist, and nodes completely outside the run (Case 4) also cost $0$ comparisons relevant to the run. Any node fully contained within the decreasing run (Case 1) will pass the decreasing skip check, bypassing the merge step, costing $0$ comparisons. Therefore, the only nodes that incur costs relevant to the decreasing run are the two Boundary nodes (Case 2) at each end of the run's contiguous interval.

    \paragraph{Left Boundary:} As defined in Lemma 6, let $v_i$ be the Case 2 node on the left boundary at level $i \ge t$, and $n_i=|v_i|$. 

    \begin{itemize}
        \item If the boundary falls in the right child $R$, applying the worst-case charge bounds the cost by $|R|-1\leq|L|=n_i - n_{i+1}$ as in Lemma 6.
        \item If the boundary falls in the left child $L$, $R$ is a pure strictly decreasing sub-run ($b_2 = \texttt{true}$) and thus cannot be non-decreasing ($b_1 = \texttt{false}$) ($|R|\geq2$ for all non-base cases).
        The algorithm utilizes \textsc{MergeFw}, exhausting the pure right half first and appending the remaining left-half run elements with $0$ comparisons. The charged cost is exactly $\vert R \vert = n_i - n_{i+1}$. 
        \item If at some level $k$ boundary perfectly aligns with the partition between $L$ and $R$, the boundary path terminates. Charging the worst-case of $1$ comparison to every run in $R$ gives $|R|=n_k-n_{k+1}$ comparisons (where $n_{k+1}=|L|$).
    \end{itemize}

    \noindent
    As the cost at any left boundary node $v_i$ is bounded by $n_i-n_{i+1}$ universally for all subcases, summing over all levels yields the telescoping sum $\sum (n_i - n_{i+1}) \le n/2^t$ from Lemma 6.

    \paragraph{Right Boundary:} By symmetry, let $v'_i$ be the right boundary node and $n'_i=|v'_i|$.

    \begin{itemize}
        \item If the boundary falls in $L$, the worst-case charge applies, identically yielding $n'_i - n'_{i+1}$.
        \item If the boundary falls in the right child $R$, the left child $L$ is a pure strictly decreasing sub-run ($a_2 = \texttt{true}$) and $R$ cannot be decreasing ($b_2 = \texttt{false}$). The algorithm invokes \textsc{MergeBw} (except when $|L|=1$; Appendix A.1 shows this does not affect the bound), exhausting the pure left half first at a cost of exactly $\vert{}L\vert{} = n'_i - n'_{i+1}$ comparisons.
        \item If the boundary aligns perfectly with the partition between $L$ and $R$ at level $k$, the boundary path terminates, and the cost is $n'_k-n'_{k+1}$ comparisons (with $n'_{k+1}=|L|$).
    \end{itemize}

    \noindent
    The same telescoping sum $\sum(n'_i - n'_{i+1}) \le n/2^t$ holds for the right boundary. Adding the costs for the left and right boundaries yields the total lower-level cost bound of $2(n/2^t)$ comparisons.
\end{proof}

\begin{theorem}[Optimized Comparison Complexity]
    An array of size $n$ containing $r_\#$ pre-existing non-decreasing or decreasing runs of lengths $r_1, r_2, \dots, r_{r_\#}$ is sorted in at most $nH + 3n$ comparisons, where the run-based entropy $H = \sum_{i=1}^{r_\#} (r_i/n) \lg \left(n/r_i\right)$.
\end{theorem}

\begin{proof}
    Combining the bounds from Lemma 15 and Lemma 16, and setting $t = \lceil \lg(n/r_i) \rceil$, the total merge comparisons charged to any decreasing run of length $r_i$ is at most $r_i \lg(n/r_i) + 2r_i$ (see proof of Lemma 7). This perfectly matches the bound for non-decreasing runs in Lemma 14. 
    
    Following the identical aggregation logic as Theorem 8, summing these charges across all $r_\#$ runs yields $\sum_{i=1}^{r_\#} \left( r_i \lg(n/r_i) + 2r_i \right) = nH + 2n$. Adding the global overhead of exactly $n-1$ comparisons for the mutually exclusive skip checks yields the final bound of $nH + 3n - 1$ comparisons.
\end{proof}

\noindent
To compute the move bound, we first separate the cost of reversals from the cost of merges.

\begin{lemma}[Cost of Reversals]
    In total, reversing decreasing subarrays makes at most $1.5n$ moves.
\end{lemma}

\begin{proof}
    According to the deferral logic in Algorithm 2, a subarray is physically reversed iff a recursive call identifies it as purely decreasing ($x_2 = \texttt{true}$). Crucially, if the reversal logic is processed, the current execution frame inevitably returns $x_2 = \texttt{false}$ to its parent. Consequently, each individual element can be involved in at most one reversal operation across the entire execution. 

    Let $m$ denote the number of elements involved in a given reversal. The cost of reversing $m$ elements is at most $1.5m$ moves, yielding an amortized bound of $1.5$ moves per element per reversal. Because every element is subject to at most one reversal, we can charge a maximal cost of $1.5$ moves to each of the $n$ elements, resulting in the bound of $1.5n$ total moves.
\end{proof}

\begin{lemma}[Optimized Moves per Run]
    For a non-decreasing or decreasing run of length $r$, the total move charge incurred is bounded by $1.5 r \lg (n/r)+4.5r$.
\end{lemma}
    
\begin{proof}
    We analyze the merge moves using the amortized accounting framework established in Section 3.2, which assigns a uniform charge of $1.5$ moves to every run element participating in an active merge (plus the global overhead of $(n-1)/2$ moves). The core structural savings of the algorithm rely on bypassing the merge step for pure nodes (Case 1):

    \begin{itemize}
        \item For a non-decreasing run, Case 1 nodes successfully pass the standard non-decreasing skip check, bypassing the merge step and incurring zero moves.
        \item For a strictly decreasing run, Case 1 nodes successfully pass the decreasing skip check, similarly bypassing the merge step and incurring zero moves.
    \end{itemize}

    \noindent
    Because the directional orientation of the elements within the run does not change the maximal capacity or structural boundaries of the nodes across the merging tree, the worst-case move charges apply identically in both regimes:
    
    \begin{itemize}
        \item At the upper levels of the tree ($h < t+1$), the total move charge incurred by the run's elements is bounded by $1.5r(t+1)$, consistent with the derivations in Lemma 9.
        \item At the lower levels of the tree ($h \ge t+1$), the sizes of boundary nodes (Case 2) down the tree form a convergent geometric series bounding moves at $3(n/2^t)$.
    \end{itemize}
    \noindent
    In total, per-run charge is $1.5r(t+1)+3(n/2^t)$ moves. By applying the convexity bound established in Theorem 11, this expression is bounded by $1.5r \log(n/r) + 4.5r$ total moves per run.
\end{proof}

\begin{theorem}[Optimized Move Complexity]
    An array of size $n$ containing $r_\#$ pre-existing non-decreasing or strictly decreasing runs of lengths $r_1, r_2, \dots, r_{r_\#}$ is sorted in at most $1.5nH + 6.5n$ moves, where the run-based entropy $H = \sum_{i=1}^{r_\#} (r_i/n) \lg \left(n/r_i\right)$.
\end{theorem}

\begin{proof}
    Paralleling the derivation in Theorem 11, summing the costs in Lemma 19 yields a total merge cost of $1.5nH+4.5n$ moves. To obtain the final move bounds, we must reintroduce all remaining global overheads. The ceiling overhead from copying halves to the buffer in the half-buffered merge adds at most $(n-1)/2$ moves (see Section 3.2). Furthermore, the deferred physical reversal of decreasing subarrays contributes at most $1.5n$ moves shown in Lemma 18. Thus, the final data movement complexity is then:
    \begin{equation*}
        1.5nH+4.5n+1.5n+(n-1)/2 \leq 1.5nH+6.5n \qedhere
    \end{equation*}
\end{proof}

\subsection{Stack Space Reduction}
The previous two algorithms utilized an explicit call stack of $O(\lg n)$ words to coordinate the processing of subproblems. However, using the static nature of merge boundaries, they can be derived on the fly with little extra information packed in $O(1)$ words by the technique from \cite[Appendix A]{XC26} (originally used to rebuild subtrees in an array-based scapegoat tree like structure). In this optimization, we utilize this technique to traverse the merging tree (in Definition 1) to coordinate merge operations of the algorithm, maintaining the following information:

\begin{itemize}
    \item The index of the current node $v$ (according to a 1-indexed binary tree), used to determine if the current node is a left or right child according to the parity of the value.
    \item The subarray boundaries of the current node in $l,r$.
    \item A bitmasked stack $\mathcal{D}$ containing the differences in sizes between the children of all ancestor nodes of $v$. One bit is sufficient per entry because Fact 2 implies the difference in size between siblings is either $0$ (when parent size is even) or $1$ (when parent size is odd). This information is used to deduce the range of sibling nodes, and then parent nodes for backtracking.
\end{itemize}

\begin{algorithm}
\caption{Node Class}
\begin{algorithmic}
\Class{Node}
    \Procedure{Node}{$n$}
        \State $v \gets 1, \ell \gets 0, r \gets n, \mathcal{D} \gets 0$ \Comment{Initialize State Variables}
        \While{$r-\ell>1$} \Call{MoveLeft}{\null} \EndWhile \Comment{Begin Postorder Traversal at Leftmost Leaf}
    \EndProcedure
    
    \Procedure{MoveLeft}{\null}
        \State $\mathcal{D}\gets(\mathcal{D}\ll1)+((r-\ell)\&1)$ \Comment{Push Children Size Difference}
        \State $v \gets v \ll 1$ \Comment{Update Node Index}
        \State$r \gets \lfloor (\ell+r)/2 \rfloor$ \Comment{Update Subarray Boundaries}
    \EndProcedure
    
    \Procedure{MoveRight}{\null}
        \State $\mathcal{D}\gets(\mathcal{D}\ll1)+((r-\ell)\& 1)$ \Comment{Push Children Size Difference}
        \State $v \gets (v \ll 1)+1$ \Comment{Update Node Index}
        \State$\ell \gets \lfloor (\ell+r)/2 \rfloor$ \Comment{Update Subarray Boundaries}
    \EndProcedure
    
    \Procedure{MovePar}{\null}
        \State $d \gets \mathcal{D} \& 1$ \Comment{Difference in Size from Sibling Node}
        \State $\mathcal{D} \gets \mathcal{D} \gg 1$ \Comment{Pop from Stack}
        \If{$v\& 1$}
            $\ell \gets \ell-(r-\ell-d)$ \Comment{Right Child, Sibling to Left Smaller by $d$}
        \Else
            \ $r \gets r+(r-\ell+d)$ \Comment{Left Child, Sibling to Right Larger by $d$}
        \EndIf
        \State $v \gets v \gg 1$ \Comment{Update Node Index}
    \EndProcedure

    \Procedure{NextPostOrder}{\null} \Comment{Standard Iterative Postorder Traversal Manager}
        \If{$v \& 1$} \Comment{Right Child: Both Subtrees Completed, Process Parent}
            \State\Call{MovePar}{\null}
        \Else \Comment{Left Child: Proceed to Right Sibling's Subtree}
            \State\Call{MovePar}{\null} \Comment{Backtrack}
            \State\Call{MoveRight}{\null} \Comment{Move to Right Subtree}
            \While{$r-\ell>1$} \Call{MoveLeft}{\null} \EndWhile \Comment{Begin at Leftmost Node of Subtree}
        \EndIf
    \EndProcedure
\EndClass
\end{algorithmic}
\end{algorithm}

\noindent
The \textsc{MoveLeft} and \textsc{MoveRight} methods are trivially correct as they simply follow the static index partitioning logic of the base algorithm. \textsc{MovePar} peeks $d \gets \mathcal{D}\&1$ to compute the sibling’s size and expand the active interval boundary accordingly. The correctness of restoring parent intervals from $\mathcal{D}$ can be shown using the symmetric inversion scheme in \cite[Appendix A]{XC26}.

Let $v,l,r,\mathcal{D}$ be in some valid initial state. We now show that \textsc{MovePar} restores the state from the left child. In the left child $\ell$ remains identical, however:
\begin{align*}
r_\text{new}&=\lfloor (\ell+r)/2 \rfloor=\ell+\lfloor (r-\ell)/2 \rfloor \\
\mathcal{D}_\text{new}&=(\mathcal{D}\ll1)+((r-\ell)\&1)
\end{align*}

\noindent
Note that $(\ell+r)\&1$ is equivalent to the difference in size between children. If we call \textsc{MovePar}, the node is a left child ($(v\&1)=0$), and the procedure executes the else block, resulting in:
\begin{align*}
    r_{\text{final}} &= r_{\text{new}} + (r_{\text{new}} - \ell + d) \\
    &= \ell+2\lfloor (r-\ell)/2 \rfloor +((r-\ell)\&1) \\
    &= \ell+(r-\ell)=r \\[0.25em]
    \mathcal{D}_\text{final}&=\mathcal{D}_\text{new}\gg 1\\
    &=((\mathcal{D}\ll1)+((r-\ell) \& 1)) \gg 1 \\
    &= \mathcal{D}
\end{align*}
\noindent
The original state is restored, hence \textsc{MovePar} correctly backtracks from the left child. By symmetry, a parallel derivation shows \textsc{MovePar} also restores the correct state from the right child.

Using the \textsc{MoveLeft}, \textsc{MoveRight} and \textsc{MovePar} move primitives, we conduct an iterative post-order traversal of the tree (via \textsc{NextPostOrder}) to process subproblems in the correct order without an explicit recursion stack. However, we still need to keep track of the return flags indicating the state of nodes (pure/impure). This is done using another bitmasked stack $\mathcal{A}$ containing the return flags of previous left children in the traversal and a single-bit variable $b$ containing the return flag of the previous right child.

\begin{algorithm}
\caption{Directional Mergesort\textsuperscript{+2}}
\newcommand{\New}{\textbf{new}\ }
\begin{algorithmic}[1]
\Function{Sort}{$A[], \ell, r, a, b$}
    \Comment{Processing Subarray $A[\ell..r)$, Children Returned $a,b$}
    \If{$r-\ell < 2$}
        \Return{\texttt{true}}
    \EndIf
    \State $m \gets \lfloor (\ell + r)/2 \rfloor$
    \If{$A[m - 1] \le A[m]$}
        \Return{$a \land b$} \Comment{Skip Check}
    \EndIf
    \If{$b$}
        \Call{MergeBw}{$A,\ell,m,r$} \Comment{Directional Merging}
    \Else
        \ \Call{MergeFw}{$A,\ell,m,r$}
    \EndIf
    \State \Return{\texttt{false}}
\EndFunction
\Procedure{DirectionalMergesort\textsuperscript{+2}}{$A[], n$}
    \State $\mathcal{A} \gets 0, b \gets 0, p \gets \New \textsc{Node}(n)$
    \While{$p.v>0$} \Comment{Repeat Until the Tree is Exhausted}
        \State$x \gets$ \Call{Sort}{$A,p.\ell,p.r,\mathcal{A} \& 1, b$} \Comment{Process Current Node}
        \If{$p.r-p.\ell>1$} $\mathcal{A}\gets \mathcal{A} \gg 1$ \Comment{Dequeue the Consumed Left Return Flag} \EndIf
        \If{$p.v \& 1$}
            $b \gets x$ \Comment{Right Child: Assign to Right Return Variable}
        \Else
            \ $\mathcal{A} \gets (\mathcal{A} \ll 1)+x$ \Comment{Left Child: Append to Left Returns Stack}
        \EndIf
        \State\Call{$p.$NextPostOrder}{\null} \Comment{Move to Next Node}
    \EndWhile
\EndProcedure
\end{algorithmic}
\end{algorithm}

\noindent
Because the parent is always visited immediately after its right child in a postorder traversal, the variable $b$ is guaranteed to hold the return value of the current node's right child (if one exists). Because the bit stack $\mathcal{A}$ contains the return flags of left subtrees along the root to $v$ path ordered by increasing depth, the top of the stack contains the deepest left subtree along the path, which is correctly the current node's left subtree. Note that the dequeue step is not executed for leaf nodes as they have no left subtree, and thus do not consume entries.

\begin{theorem}
    Directional Mergesort\textsuperscript{+2} sorts an array of size $n$ with run-based entropy $H$ (defined for exclusively non-decreasing runs) using no more than $nH + 3n$ comparisons and $1.5nH + 5n$ element moves. Excluding the merge buffer of size $\lceil n/2 \rceil$, the algorithm uses $O(1)$ auxiliary words.
\end{theorem}

\begin{proof}
    The optimized version processes all subproblems in postorder identically to the base variant, maintaining the correct return flags, resulting in identical behavior to the base variant (making the same comparisons and moves). Hence, the comparison and move bounds in Corollary 12 is retained.

    However, the total auxiliary space usage, excluding the merge buffers, is reduced to $O(1)$ words: the algorithm utilizes a fixed number of single word variables alongside the two bit stacks $\mathcal{A}$ and $\mathcal{D}$. Because the height of the merge tree is bounded by $\lceil \lg n \rceil$, $\mathcal{D}$ and $\mathcal{A}$ contain at most $\lceil \lg n \rceil$ single-bit entries. According to the standard assumption that words are at least $\lg n$ bits wide, each bit stack can fit within a single word.

    The arithmetic operations for managing the tree traversal is $O(n)$. This is because an iterative postorder traversal in a tree of $2n-1$ nodes calls the move primitives no more than $2((2n-1)-1)$ times, as each edge is walked at most twice. Each move primitive makes $O(1)$ arithmetic operations, totaling $O(n)$ arithmetic operations across the algorithm's execution. Consequently, the overall time complexity of the algorithm remains bounded by $O(nH+n)$.
\end{proof}

\subsection{Summary}
\begin{corollary}
    Directional Mergesort\textsuperscript{++} sorts an array of size $n$ with run-based entropy $H$ (defined for both non-decreasing and decreasing runs) using no more than $nH + 3n$ comparisons and $1.5nH + 6.5n$ moves. Excluding the merge buffer of size $\lceil n/2 \rceil$, the algorithm uses $O(1)$ auxiliary words.
\end{corollary}

\begin{proof}
    By a parallel application of stack space reduction optimization to Directional Mergesort\textsuperscript{+1}, the resulting algorithm operating within $O(1)$ words inherits the bounds in Theorem 17 and Theorem 20. The bit stack $\mathcal{A}$ and variable $b$ is expanded to accommodate two-bit entries, requiring twice the space, but this does not violate $O(1)$ word space constraints.
\end{proof}

\begin{remark}
    Whilst Corollary 22 establishes an upper bound of $nH + 3n$ comparisons for simplicity, a tighter analysis of the right boundary accounting for the size 2 base case yields a stronger bound of $nH + 3n - r_\#$. The full proof is deferred to Appendix A.2.
\end{remark}

\section{Concluding Remarks}
In this paper, we introduced directional mergesort, a minimalistic adaptive sorting algorithm that matches the optimal $nH + 3n$ comparisons bound set by state-of-the-art natural mergesorts Powersort and Peeksort. Unlike predecessors that rely on dynamic run-scanning and variable merge tree geometries, our approach operates strictly over static index partitions using localized skip checks and dynamic merge-direction choices. This ensures the merging tree structure remains completely static, simplifying the underlying element-charging analysis across tree levels.

Furthermore, directional mergesort\textsuperscript{++} then leverages this data oblivious merge tree geometry to compress both child-size differences and purity return flags into single word bit stacks, reducing auxiliary space usage (aside from the merge buffer) to $O(1)$ words. This demonstrates that optimally run-adaptive sorting does not require $O(\lg n)$ words for boundary tracking or intrusion into the data itself to simulate this space. Furthermore, by tracking non-decreasing and strictly decreasing runs using two flags, the algorithm adapts to both orientations of runs without incurring any additional overhead in comparisons, matching the decreasing run adaptivity of natural mergesorts.

When considering lower-order terms, a more careful analysis in Appendix A.2 gives the tighter bound of $nH+3n-r_\#$ comparisons for directional mergesort\textsuperscript{++}, exactly matching the subtractive $-r_\#$ term from Powersort. However, the proven comparison bound of the algorithm remains slightly worse than that of Peeksort, which has a subtractive term of $-2r_\#$.

Practical implementations of the algorithm could consider switching to a modified insertion sort that keeps track of non-decreasing or decreasing purity for small subarrays ($r - \ell \le \text{min\_run}$). Additionally, optimized merge procedures such as galloping during merges could accelerate unbalanced run comparisons, while ping-pong merging from \cite{CG14} could reduce element moves from $1.5nH + O(n)$ to $nH + O(n)$ (using half-buffered merges for the final pass to avoid a full size $n$ buffer).

Whilst directional mergesort has optimal worst-case comparison bounds, its static merge steps may cause worse performance in the average case compared to other run-adaptive algorithms, leading to inferior performance in practice for typical inputs. We leave the empirical validation of its practical performance against production implementations like Timsort to future work.

A potential practical application of directional mergesort\textsuperscript{++} involves substituting the linear-time stable in-place merge algorithm from \cite{KK08} for the half-buffered merge. This removes the need for a merge buffer of size $\lceil n/2 \rceil$, reducing total auxiliary space usage to $O(1)$ words. However, the large constant overhead (in comparisons and moves) of the in-place merge procedure causes us to lose our exact bounds of $nH+3n-r_\#$ comparisons and $1.5nH+6.5n$ moves. However, the running time of merges remains linear ($\Theta(m)$ to merge $m$ elements in total), so the overall time complexity of $O(nH+n)$ is retained. Whilst this combination was achieved by a recent result in \cite{GGS26}, their approaches incur significant overhead to avoid explicitly storing a stack of $O(\lg n)$ words, involving repeated array inspections or complex encoding techniques. This overhead is completely avoided by directional mergesort\textsuperscript{++}, whose static merging tree structure inherently supports the reduction to $O(1)$ words at the cost of only arithmetic operations.

Finally, we note that if the desired bound is only $nH+O(n)$ comparisons and $O(nH+n)$ running time ($H$ defined for exclusively non-decreasing runs), a simple mergesort with skip checks only is already sufficient. The proof mirrors Section 3.2, hence it is left to Appendix B.

\section*{Acknowledgements}
The authors would like to thank Jay Llanes for raising the question of whether a simple mergesort with skip checks could achieve $O(nH + n)$ complexity, which initiated this work.

AI tools were used solely for standard copy-editing tasks, such as language refinement and formatting. The technical aspects, including all algorithms and proofs, was formulated entirely by the authors, who assume responsibility for whole content.

\appendix
\newpage
\section{Tightened Boundary Analysis}
In this Appendix, we provide the complete analysis of directional mergesort\textsuperscript{++}.

\subsection{Directional Merging Savings}
This subsection analyzes the correctness of the directional merge for all possible run configurations in directional mergesort\textsuperscript{+1}, managed by the \textit{backward condition} on line 15 of Algorithm 2. As directional merge savings only apply to boundary nodes where the run contains an entire half and extends into the opposite half (Subcase B), we will only consider such situations. We otherwise count no directional merging savings and any merge direction is valid. We begin by considering non-decreasing runs relative to a boundary node $v$ with left child $L$ and right child $R$:
\begin{itemize}
    \item \textbf{Case i (Prefix):} A non-decreasing run fully contains $L$ so $a_1=\texttt{true}$. As the end of the non-decreasing run is in $R$, it cannot be pure non-decreasing (as the end of a non-decreasing run is decreasing) hence $b_1=\texttt{false}$. Then, the backward condition evaluates to $b_1 \lor \lnot( b_2 \lor a_1) =\texttt{false}\lor\lnot(b_2 \lor \texttt{true})=\texttt{false}$, so the algorithm chooses \textsc{MergeFw}.
    \item \textbf{Case ii (Suffix):} The run fully contains $R$, and extends leftwards into $L$. Then $b_1=\texttt{true}$, ensuring the backward condition evaluates to \texttt{true}. The algorithm calls \textsc{MergeBw}.
\end{itemize}
\noindent
As discussed in Section 2, these are the correct merge directions that generate directional merging savings for non-decreasing runs. We apply a mirrored argument for strictly decreasing runs:
\begin{itemize}
    \item \textbf{Case i (Prefix):} The decreasing run fully contains $L$ ($a_2=\texttt{true}$) and extends rightwards into $R$, so $b_2=\texttt{false}$ (as $R$ contains the end of a decreasing run which is non-decreasing). The analysis splits into two further cases:
    \begin{itemize}
        \item \textbf{Subcase X ($|L|>1$):} As a sequence strictly longer than $1$ element cannot be both non-decreasing and decreasing simultaneously, $a_2=\texttt{true}$ implies $a_1=\texttt{false}$. Hence, the backward condition evaluates to $b_1 \lor \lnot( b_2 \lor a_1) =b_1\lor\lnot(\texttt{false} \lor \texttt{false})=b_1\lor \texttt{true}=\texttt{true}$ so the algorithm chooses \textsc{MergeBw}. The merge then exhausts all elements in $L$ before selecting any run elements from $R$ (as established in Section 4.1).
        \item \textbf{Subcase Y ($|L|=1$):} Merging occurs only when $\vert{}R\vert{}=2$ (and $m=|L|+|R|=3$), as $\vert{}R\vert{}=1$ is handled by base cases. We are interested only when the decreasing prefix has size $2$, spanning $1$ element in $L$ and $1$ in $R$ ($k_1>k_2 \leq k_3$) as a $1$-element run falls under Subcase C and a $3$-element run bypasses the merge entirely. Because $|L|=1$, the base case returns $a_1=\texttt{true}$ and the backward condition $b_1 \lor \lnot(b_2 \lor a_1)$ simplifies to $b_1$, leaving the merge direction unknown. Nevertheless, merging $m=3$ elements requires at most $m-1=2$ comparisons regardless of whether \textsc{MergeBw} or \textsc{MergeFw} is called. The worst-case analysis of Lemma 16 charges at least $2$ comparisons in this scenario: $1$ to the run containing $k_1k_2$ in Subcase B (or $2$ at upper levels) and $1$ to $k_3$ in Subcase A (or upper levels); the upper bound remains valid regardless of merge direction.
    \end{itemize}
    \item \textbf{Case ii (Suffix):} The decreasing run fully contains $R$ ($b_2=\texttt{true}$) and extends leftwards into $L$. Because $r-\ell\leq2$ is handled by base cases, we deduce $r-\ell\geq3$, and $|R|=\lceil (r-\ell)/2\rceil\geq2$ according to the index partitioning logic of the algorithm. Because a sequence of size $2$ or greater cannot be both non-decreasing and decreasing simultaneously, $b_2=\texttt{true}$ implies $b_1=\texttt{false}$. The backward condition then evaluates to $b_1 \lor \lnot( b_2 \lor a_1) =\texttt{false}\lor\lnot(\texttt{true} \lor a_1)=\texttt{false}$, thus the else branch executes correctly utilizing \textsc{MergeFw} to exhaust $R$ first.
\end{itemize}

\noindent
Hence, the directional merging logic of Algorithm 2 successfully captures directional savings for decreasing runs where applicable, ensuring the comparison charging (in Lemma 16) remains valid.

\subsection{Tightened Bound}
Here we prove the refined bound for Corollary 22 noted in Remark 23. This is done by a more precise analysis of the boundary path.

\begin{theorem}[Tighter Comparison Bound]
    Directional mergesort\textsuperscript{++} sorts an array of size $n$ containing $r_{\#}$ runs in at most $nH+3n-r_{\#}$ comparisons (where $H$ is defined for both non-decreasing and decreasing runs).
\end{theorem}

\begin{proof}
    Recall that sequences of size 2 are resolved using the skip check comparison during the base case (see Algorithm 2), bypassing the merge step entirely and incurring 0 merge comparisons. Since the skip check comparisons are accounted for separately in the global $n-1$ overhead bound, we can refine our boundary node analysis when $r \geq 2$ ($r=1$ handled separately shortly).

    We redefine the termination condition of the right boundary path at a level $k'$ to be when:
    \begin{itemize}
        \item The right boundary of the run aligns with the partition between children (Subcase C).
        \item Or the size of the next right boundary node $\vert{}v'_{k'+1}\vert{} = 2$. Additionally, if the first boundary node at level $t$ already has size $2$, the boundary path is empty. As a boundary node of size 2 incurs no merge comparisons, and its children are leaves of size 1 which return instantly, no costs are incurred below and the early termination already accounts for all possible charges.
    \end{itemize}

    \noindent
    We examine the final size term $n'_{k'+1}$ for the right boundary:
    \begin{itemize}
        \item \textbf{Case i (Standard Termination):} The boundary perfectly aligns (Subcase C) at some level $k'$ where $n'_{k'} \geq 3$ then $n'_{k'+1} = |R| = \lceil n'_{k'}/2 \rceil \geq 2$.
        \item \textbf{Case ii (Base Case Termination):} The path terminates due to the redefined condition, thus $n'_{k'+1} = 2$.
    \end{itemize}

    \noindent
    These two cases are fully representative of all possible boundary path terminations. This is because a node of size $2$ cannot terminate the boundary path as the new condition terminates the path immediately if the next node is of size $2$. Furthermore, nodes of size $1$ are either purely run elements or non-run elements, so they can only be Case 1 or Case 4, and cannot be Case 2, hence they cannot be part of the boundary path. Thus, the lower bound $n'_{k'+1}\geq2$ holds for the right boundary.

    Summing the charges $n'_i - n'_{i+1}$ from level $t$ to level $k'$ yields a telescoping sum bounded by:

    $$\sum_{i=t}^{k'} (n'_i - n'_{i+1}) = n'_t - n'_{k'+1} \le \lceil n/2^t \rceil - 2 \leq n/2^t-1$$

    \noindent
    If the boundary path is empty, the charge is then $0\leq n/2^t-1$ (as $t\leq \lceil \lg (n/2) \rceil\leq \lg n$ when $r\geq2$). Combining with the previous left boundary cost of $n/2^t$, the total lower level charge is $2(n/2^t)-1$ comparisons. Integrating this tighter deduction with the upper level costs of $rt$ (in Lemma 14 and Lemma 15) and applying the convexity bound yields:
    $$r \lg(n/r) + 2r - 1$$
    \noindent
    \textbf{Case $r = 1$:} A run of length $r = 1$ consists of a single element at a leaf node in the merging tree. Since the depth of any leaf in a binary merging tree of size $n$ is at most $\lceil \lg n \rceil$, charging 1 comparison to the single run element for each ancestor of the leaf node yields:
    $$\lceil \lg n \rceil \le \lg n + 1 = 1  \lg(n/1) + 2(1) - 1 = r \lg(n/r) + 2r - 1$$

    \noindent
    Thus, the bound of $r\lg (n/r)+2r-1$ holds for all runs. Summing across all $r_{\#}$ runs  and reintroducing the global $n-1$ skip checks gives the final bounded comparisons:
    \begin{equation*} nH + 2n - r_{\#} + (n - 1) \le nH + 3n - r_{\#} \qedhere \end{equation*}
\end{proof}

\newpage
\section{Simple Adaptive Mergesort Analysis}
This Appendix is dedicated to answering the question of Jay Llanes, regarding the well-known simple adaptive mergesort with skip checks:
\begin{algorithm}
\caption{Simple Adaptive Mergesort}
\begin{algorithmic}[1]
\Procedure{Sort}{$A[], \ell, r$}
    \Comment{Processing Subarray $A[\ell..r)$}
    \If{$r-\ell < 2$}
        \Return
    \EndIf
    \State $m \gets \lfloor (\ell+r)/2 \rfloor$
    \State \Call{Sort}{$A,\ell,m$}
    \State \Call{Sort}{$A,m,r$}
    \If{$A[m - 1] \leq A[m]$}
        \Return \Comment{Skip Check}
    \EndIf
    \State\Call{MergeFw}{$A,\ell,m,r$}
\EndProcedure
\end{algorithmic}
\end{algorithm}

\noindent
Unlike directional mergesort (Algorithm 1), Algorithm 5 uses only standard forward merging (\textsc{MergeFw}) rather than dynamic directional selection. Consequently, while Case 1 (Fully Contained) nodes bypass the merge entirely via line 6 and cost $0$ comparisons, Case 2 (Boundary) nodes cannot leverage the directional savings established in Section~3.1. 

To account for this, we adopt the flat amortized charging scheme introduced in Section~3.2: every element involved in an active merge loop is charged 1 comparison per level regardless of local savings, while Case 1 nodes and Case 4 nodes incur 0 cost as a whole.

\subsection{Comparison Complexity}

As established in Lemma~4, the global overhead of skip checks across the $n-1$ internal nodes of the binary merging tree contributes exactly $n-1$ comparisons. We isolate this overhead first and analyze the per-run merge comparison charges over the two tree depth regimes using threshold depth $t+1 = \lceil \lg(n/r) \rceil+1$ for a run of length $r$.

\begin{lemma}[Upper-Level Comparison Charge]
For a non-decreasing run of length $r$, the total merge comparison charge incurred across upper levels ($h < t+1$) is bounded by $r(t+1)$.
\end{lemma}
\begin{proof}
    Charging each of the $r$ run elements at most 1 comparison per level across $t+1$ levels yields an upper-level cost bounded by $r(t+1)$.
\end{proof}

\begin{lemma}[Lower-Level Comparison Charge]
For a non-decreasing run of length $r$, the total merge comparison charge incurred across lower levels ($h \ge t+1$) is bounded by $2(n/2^t)$.
\end{lemma}
\begin{proof}
At levels $h \ge t+1>t$, node sizes are bounded by $r$, so Case 3 nodes cannot exist. Furthermore, Case 1 and Case 4 nodes contribute 0 charge to the run. Thus, non-zero charges arise solely from at most two Case 2 boundary nodes per level. 

By Fact~2, a boundary node $v$ at level $h$ has size at most $\lceil n/2^h \rceil$. Since the run cannot fully contain a boundary node, $v$ contains at most $\lceil n/2^h \rceil - 1 \le n/2^h$ run elements. Across both boundary nodes, at most $2(n/2^h)$ run elements are charged 1 comparison at level $h$. Summing this geometric series across all lower levels $h \ge t+1$ yields:
\begin{equation*} \sum_{h=t+1}^{\lg n} 2(n/2^h) \le 2 \cdot \frac{n/2^{t+1}}{1 - 1/2} = 4(n/2^{t+1})=2(n/2^t) \qedhere \end{equation*}
\end{proof}

\begin{theorem}[Comparison Complexity of Simple Adaptive Mergesort]
An array of size $n$ containing $r_\#$ non-decreasing runs of lengths $r_1, r_2, \dots, r_{r_\#}$ is sorted by Simple Adaptive Mergesort using at most $nH + 4n$ comparisons, where $H = \sum_{i=1}^{r_\#} (r_i/n) \lg(n/r_i)$.
\end{theorem}
\begin{proof}
Combining upper and lower level comparison charges for a run of length $r$ yields:
$$r(t+1) + 2(n/2^t)$$
Letting $\gamma = \lceil \lg(n/r) \rceil - \lg(n/r) \in [0, 1)$, we rewrite this bound as:
$$r \left(\lg(n/r) + \gamma +1\right) + 2\left(r/2^\gamma\right) = r \lg(n/r) + r\left(1+\gamma + 2\left(2^{-\gamma}\right)\right)$$
Applying the secant line bound to the convex expression $2^{-\gamma}$ for $\gamma \in [0,1]$:
\begin{align*}
    2^{-\gamma} &\leq 1-(1/2)\gamma \\[0.25em]
    2(2^{-\gamma}) &\leq 2-\gamma \\[0.25em]
    2(2^{-\gamma})+\gamma &\leq 2
\end{align*}
Thus, the total comparison charge per run is at most $r \lg(n/r) + 3r$. Summing across all $r_\#$ runs:
$$\sum_{i=1}^{r_\#} \left( r_i \lg(n/r_i) + 3r_i \right) = nH + 3n$$
Finally, reintroducing the $n-1$ skip check comparisons gives:
\begin{equation*} nH + 3n + (n - 1) \le nH + 4n \qedhere \end{equation*}
\end{proof}

\subsection{Data Movement Complexity}

\begin{corollary}[Move Complexity of Simple Adaptive Mergesort]
    An array of size $n$ containing $r_\#$ non-decreasing runs of lengths $r_1, r_2, \dots, r_{r_\#}$ is sorted by Simple Adaptive Mergesort using at most $1.5nH + 5n$ moves, where $H = \sum_{i=1}^{r_\#} (r_i/n) \lg(n/r_i)$.    
\end{corollary}

\begin{proof}
    As Theorem 11 already utilizes an amortized accounting scheme that charges a flat $1.5$ moves per element oblivious to potential localized directional merging savings, the bound established in Theorem 11 holds for the simplified version.
\end{proof}

\end{document}